\documentclass[12pt]{article}

\usepackage[usenames]{color}
\usepackage{amssymb}
\usepackage{graphicx}
\usepackage{amscd}
\usepackage{tabularx}
\usepackage{lscape}

\usepackage[colorlinks=true,
linkcolor=webgreen,
filecolor=webbrown,
citecolor=webgreen]{hyperref}

\definecolor{webgreen}{rgb}{0,.5,0}
\definecolor{webbrown}{rgb}{.6,0,0}

\usepackage{color}
\usepackage{fullpage}
\usepackage{float}

\usepackage{graphics,amsmath,amssymb}
\usepackage{amsthm}
\usepackage{amsfonts}
\usepackage{latexsym}
\usepackage{epsf}

\DeclareMathOperator{\crep}{crep}
\DeclareMathOperator{\prefgeab}{prefgeab}
\DeclareMathOperator{\prefgtab}{prefgtab}
\DeclareMathOperator{\prefeqab}{prefeqab}

\DeclareMathOperator{\facgeab}{facgeab}
\DeclareMathOperator{\facgtab}{facgtab}
\DeclareMathOperator{\faceqab}{faceqab}
\DeclareMathOperator{\facab}{facab}
\DeclareMathOperator{\facsmallab}{facsmallab}
\DeclareMathOperator{\faclargeab}{faclargeab}

\DeclareMathOperator{\conj}{conj}
\DeclareMathOperator{\cexp}{ce}
\DeclareMathOperator{\ccexp}{cce}
\DeclareMathOperator{\lcce}{lcce}
\DeclareMathOperator{\gcce}{gcce}
\DeclareMathOperator{\ace}{ace}

\def\codelink{\url{https://cs.uwaterloo.ca/~shallit/papers.html}}

\author{Jeffrey Shallit and Ramin Zarifi \\
School of Computer Science \\
University of Waterloo \\
Waterloo, ON  N2L 3G1 \\
Canada \\
{\tt shallit@uwaterloo.ca} \\
{\tt rzarifi@edu.uwaterloo.ca}}

\title{Circular critical exponents for Thue-Morse factors}

\begin{document}

\maketitle

\begin{abstract}
We prove various results about the largest exponent of a repetition
in a factor of the Thue-Morse word, when that factor is considered
as a circular word.  Our results confirm and generalize previous
results of Fitzpatrick and Aberkane \& Currie.  
\end{abstract}

\theoremstyle{plain}
\newtheorem{theorem}{Theorem}
\newtheorem{corollary}[theorem]{Corollary}
\newtheorem{lemma}[theorem]{Lemma}
\newtheorem{proposition}[theorem]{Proposition}

\theoremstyle{definition}
\newtheorem{definition}[theorem]{Definition}
\newtheorem{example}[theorem]{Example}
\newtheorem{conjecture}[theorem]{Conjecture}

\theoremstyle{remark}
\newtheorem{remark}[theorem]{Remark}

\section{Introduction}

Consider the English
word {\tt amalgam}; it has a factor\footnote{A {\it factor} is a contiguous
block lying inside another word.} {\tt ama} of period $2$ and length
$3$, so we can consider {\tt ama} to be a $3\over 2$ power.  However, if we think
of {\tt amalgam} as a ``circular word'' or ``necklace'', where the word ``wraps around'',
then it has the factor {\tt amama} of period $2$ and length $5$.
We say that {\tt amalgam} has a circular critical exponent of $5 \over 2$.

The famous Thue-Morse infinite word 
$${\bf t} = t_0 t_1 t_2 \cdots = {\tt 01101001} \cdots $$
has been studied extensively since its introduction by
Thue in 1912 \cite{Thue:1912,Berstel:1992}.
In particular, Thue proved that the largest repetitions
in $\bf t$ are $2$-powers (also called ``squares'').  

It was only fairly recently,
however, that the repetitive properties of its factors, 
{\it considered as circular words}, have been studied.  Fitzpatrick \cite{Fitzpatrick:2005} showed that, for all
$n \geq 1$, there is a length-$n$ factor of $\bf t$ with circular critical
exponent $< 3$.  Aberkane and Currie \cite{Aberkane&Currie:2004} conjectured that
for every $n \geq 1$, some length-$n$ factor of $\bf t$ has circular critical
exponent $\leq {5 \over 2}$, and, using a case analysis, they later proved this conjecture
\cite{Aberkane&Currie:2005a}.

In this paper we show how to obtain the Aberkane-Currie result, and much more,
using an approach based on first-order logic and the {\tt Walnut} prover, written
by Hamoon Mousavi.

\section{Basics}

The $i$'th letter of a word $w$ is written $w[i]$.  The notation
$w[i..j]$ represents the word 
$$w[i] w[i+1] \cdots w[j].$$
If $i \geq j+1$, then $w[i..j] = \epsilon$, the empty word.

An infinite (resp., nonempty finite) word $w$ has a period $p \geq 1$ if
$w[i] = w[i+p]$ for all $i \geq 0$ (resp., all $i$ with $0 \leq i <
|w|-p$).  For finite words of length $n$, we restrict our attention to 
periods that are $\leq n$.
A word can have multiple periods; for example, the 
English word
{\tt alfalfa} has periods $3, 6$, and $7$.
The smallest period is called {\it the\/} period and is denoted
$p(w)$.  The
{\it exponent} of a finite word $w$ is defined to be $\exp(w) = |w|/p(w)$;
it measures the largest amount of (fractional) repetition of a word.
The period of 
{\tt alfalfa} is $3$, and it has length $7$; hence
its exponent is ${7 \over 3}$.

A word is called a {\it square} if its exponent is $2$.  If its exponent
is greater than $2$, it is called an {\it overlap}.  Thus, for example,
the English word {\tt murmur} is a square and the French word {\tt entente} is an overlap.

The {\it critical exponent} of a word $w$ is the supremum, over all
finite nonempty factors $x$ of $w$, of $\exp(x)$; it is denoted
$\cexp(w)$.  For example, {\tt Mississippi} has critical
exponent $7/3$, arising from the overlap {\tt ississi}.   

We can also define this notion for ``circular words'' (aka ``necklaces'').
We say two words $x,y$ are {\it conjugate} if one is a cyclic shift of the 
other; alternatively, if there exist (possibly empty) words $u, v$ such that $x = uv$ and $y = vu$.  For example, the English words {\tt listen} and
{\tt enlist} are conjugates.

We let $\conj(w)$ denote the set of all cyclic shifts of $w$:
$$ \conj(w) = \{ yx \ : \ \exists x, y \text{ such that } w = xy \}.$$
For example, the conjugates of {\tt ate} are $\{ {\tt ate}, {\tt tea}, {\tt eat} \}$.

Here is the most fundamental definition of our paper:
\begin{definition}
The {\it circular critical exponent} of a word $w$, denoted by $\ccexp(w)$,
is the supremum of $\exp(x)$ over all finite nonempty factors $x$
of all conjugates of $w$.  
\end{definition}

Note that $\ccexp (w)$ can be as much as twice as large as $\cexp(w)$.
See \cite{Mousavi&Shallit:2013} for more about this notion for infinite
words.

\subsection{The Thue-Morse word}

The Thue-Morse word $\bf t$
has many equivalent definitions \cite{Allouche&Shallit:1999},
but for us it will be sufficient
to describe it as the fixed point, starting with $\tt 0$, of
the morphism $\mu$ mapping ${\tt 0} \rightarrow {\tt 01}$ and
${\tt 1} \rightarrow {\tt 10}$.

A basic fact about the binary alphabet is that every word of length
$\geq 4$ has critical exponent at least $2$.
Thue proved that the Thue-Morse word has no overlaps.  Thus we get
the following (trivial) result about factors of the Thue-Morse word.

\begin{proposition}
Let $x$ be a nonempty factor of the Thue-Morse word.  Then
$\cexp(x) \in \{ 1, {3 \over 2}, 2 \}$.  Furthermore,
$\cexp(x) = 2$ if $|x| \geq 4$.
\label{prop1}
\end{proposition}

In this paper, we prove the analogue of Proposition~\ref{prop1} for
the circular critical exponent.  Here the statement is more complicated
and the analysis more difficult.

\subsection{{\tt Walnut}}

Our main software tool is the {\tt Walnut} prover, written by 
Hamoon Mousavi \cite{Mousavi:2016}.  This Java program deals with
deterministic finite automata with output (DFAO's) and
$k$-automatic sequences $(a_n)_{n \geq 0}$.  A $k$-DFAO is a finite-state machine
$M = (Q, \Sigma_k, \delta, q_0, \Delta, \tau)$, where $Q$ is a finite nonempty
set of states, $\Sigma_k = \{0,1,\ldots, k-1\}$ is the input alphabet,
$\delta:Q \times \Sigma_k \rightarrow Q$ is the transition function
(which is extended to $Q \times \Sigma_k^*$ in the obvious way),
$q_0$ is the initial state, $\Delta$ is the output alphabet, and
$\tau:Q \rightarrow \Delta$ is the output mapping.  DFAO's are an obvious
generalization of ordinary DFA's.  A sequence $(a_n)_{n \geq 0}$ is said
to be {\it computed by the $k$-DFAO $M$\/}
if $\tau(\delta(q_0, (n)_k)) = a_n$, where $(n)_k$ denotes
the base-$k$ representation of $n$.  (Unless otherwise stated, we assume
that all automata read the base-$k$ representation of $n$ from left to right,
starting with the most significant digit.)  If a sequence $(a_n)_{n \geq 0}$ is
computed by a $k$-DFAO, it is said to be {\it $k$-automatic}.

{\tt Walnut} can evaluate
the truth of a first-order statement $S$ involving indexing of $k$-automatic sequences,
logical connectives, and quantifiers $\exists$ and $\forall$.  If there are free
variables, it produces an automaton accepting the base-$k$ representation of 
the values of the free variables for which $S$ evaluates to true.  One minor
technical point is that the automata it produces give the correct answer,
even when the input is prefixed by any number of leading zeroes.

The syntax of {\tt Walnut} statements is more or less self-explanatory.
The interested reader can enter the {\tt Walnut} commands we give and
directly reproduce our results.

All computations in this paper, unless otherwise indicated, were performed
on an Apple MacBook Pro with 16 GB of memory,
running macOS High Sierra, version 10.13.3.    All the code we discuss
is available for download at \codelink \, .  For the Thue-Morse word,
the computations all run in a matter of seconds.

\subsection{Minimality}

There is a notion of minimality for DFAO's that exactly parallels the notion of
ordinary DFA's.  We say that two states $p,q$ of a $k$-DFAO are {\it distinguishable\/}
if there exists a string $x \in \Sigma_k^*$ such that $\tau(\delta(p,x)) \not=
\tau(\delta(q,x))$.  Then the analogue of the Myhill-Nerode theorem for 
DFAO's is the following, which is easily proved:
\begin{proposition}
There is a unique minimal $k$-DFAO equivalent to any given $k$-DFAO.  Furthermore,
a $k$-DFAO $M$ is minimal iff (a) every state of $M$ is reachable from the start
state and (b) every pair of distinct states is distinguishable.
\end{proposition}
\noindent We observe that the automata that {\tt Walnut} computes are guaranteed to
be minimal.

We will need the following lemma.

\begin{lemma}
Let $M_1 = (Q_1, \Sigma_k, \delta_1, q_1, \Delta_1, \tau_1)$ and
$M_2 = (Q_2, \Sigma_k, \delta_2, q_2, \Delta_2, \tau_2)$ be two minimal DFAO's.
Let 
$$ M = (Q, \Sigma_k \times \Sigma_k, \delta, q_0, \Delta_1 \times \Delta_2, \tau)$$
be the cross product automaton defined by
\begin{itemize}
\item $Q = Q_1 \times Q_2$;
\item $\delta([p,q], a) = [\delta_1(p,a), \delta_2(q,a)]$;
\item $q_0 = [q_1, q_2]$;
\item $\tau([p,q]) = [\tau_1(p), \tau_2 (q)]$.
\end{itemize}
Then every pair of distinct states of $M$ is distinguishable.
\end{lemma}

\begin{proof}
Let $[p,q]$ and $[p',q']$ be two distinct states of $M$.  Without loss of generality,
assume $p \not= p'$.  Then, since $M_1$ is minimal, we know that $p$ and $p'$
are distinguishable, so there exists $x$ such that $\tau_1(\delta_1(p,x)) \not=
\tau_1(\delta_1(p',x))$.   Then $\tau(\delta([p, q],x)) = [\tau_1(\delta_1(p,x)),
\tau_2(\delta_2(q,x))] \not= [\tau_1 (\delta_1(p',x), \tau_2(\delta_2(q',x))] =
\tau(\delta([p',q'],x)) $.  So $[p,q]$ and $[p', q']$ are distinguishable by $x$.
\end{proof}

\begin{corollary}
Let $M_1$ and $M_2$ be minimal $k$-DFAO's.  Form their cross product automaton, and
remove all states unreachable from the start state.  The result is minimal.
\label{cross}
\end{corollary}

Corollary~\ref{cross} gives a way to form the minimal cross product automaton,
but in practice we can do something even more efficient:  namely, using
a breadth-first approach, we can start from the start state $[q_1, q_2]$
and incrementally add only those states reachable from it.  

\section{First-order formulas for factors}

We start by developing a useful first-order logical formula with
free variables $i,m,n,p,s$.
We want it to assert that
\begin{align}
& \text{in the circular word given by the length-$n$ word starting } 
\nonumber \\
& \text{at position $s$ in the Thue-Morse word, there is a factor} 
\label{log} \\
&  \text{$w$ of length $m$ and (not necessarily least) period $p \geq 1$ starting at position $i$.}
\nonumber
\end{align}

In order to do this, we will conceptually
repeat the word $x = {\bf t}[s..s+n-1]$ twice, as
depicted below, where the black vertical line separates the two
copies.  The factor $w$ is indicated in grey; it may or may not straddle
the boundary between the two copies.
\begin{figure}[H]
\begin{center}
\includegraphics[width=5in]{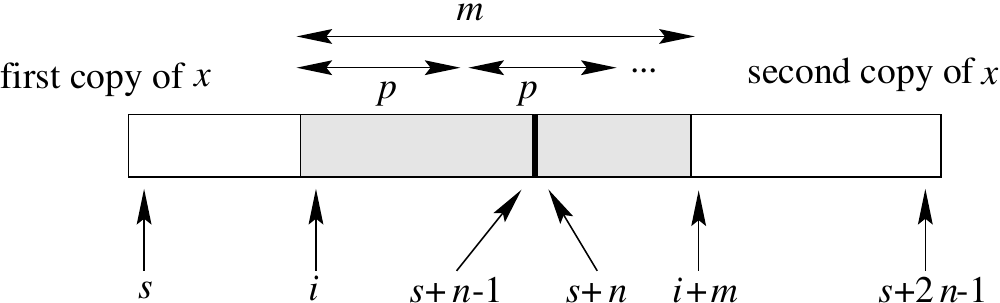}
\end{center}
\caption{Factor of a circular word of length $n$}
\end{figure}
Here indices should be interpreted as ``wrapping around''; the index
$s+n+j$ is the same as $s+j$ for $0 \leq j <n$.   Then the assertion
that $w$ has period $p$ potentially corresponds to three different ranges
of $j$:
\begin{itemize}
\item Both $j$ and $j+p$ lie in the first copy of $x$, so
we compare ${\bf t}[j]$ to ${\bf t}[j+p]$ for all $j$ in this range:
$i \leq j < \min(s+n-p,i+m-p)$.

\item $j$ lies in the first copy of $x$, but $j+p$ lies in the
second copy, so  we compare ${\bf t}[j]$ to ${\bf t}[j+p-n]$ 
for all $j$ in this range:
$\max(i,s+n-p) \leq j < \min(s+n,i+m-p)$. 

\item Both $j$ and $j+p$ lie in the second copy of $x$, so
we compare ${\bf t}[j-n]$ to ${\bf t}[j+p-n]$ for all $j$ in this range:
$\max(i,s+n) \leq j < i+m-p$.
\end{itemize}

Putting this all together, we get the following logical formula
that asserts the truth of statement \eqref{log}:
\begin{align*}
& \crep(i,m,n,p,s) :=  \\
& (\forall j  \, ((j \geq i) \, \wedge \, (j<s+n-p) \, \wedge\,  (j<i+m-p)) 
\implies {\bf t}[j] = {\bf t}[j+p] ) \, \wedge \\
& (\forall j \,  ((j \geq i) \, \wedge \, (j<s+n) \, \wedge \, (j \geq s+n-p) \, \wedge\, (j<i+m-p)) \implies {\bf t}[j] = {\bf t}[j+p-n]) \wedge \\
& (\forall j \, ((j \geq i) \, \wedge \, (j\geq s+n) \, \wedge \, (j<i+m-p)) \implies {\bf t}[j-n] = {\bf t}[j+p-n] ) 
\end{align*}

The translation into {\tt Walnut} is as follows:

{\tt
\begin{verbatim}
def crep "(Aj ((j>=i)&(j+p<s+n)&(j+p<i+m)) => T[j]=T[j+p]) &
     (Aj ((j>=i)&(j<s+n)&(j+p>=s+n)&(j+p<i+m)) => T[j]=T[(j+p)-n]) &
     (Aj ((j>=i)&(j>=s+n)&(j+p<i+m)) => T[j-n]=T[(j+p)-n])":
\end{verbatim}
}


The resulting automaton implementing $\crep(i,m,n,p,s)$ has 1423 states.  Note
that our formula does not impose conditions such as
$p \geq 1$ or $p \leq n$ or  $m \leq n$, which are
required for $\crep$ to make sense. These conditions (or
stronger ones that imply them) must be included in any predicate
that makes use of $\crep$.  Neither does the predicate assert that the
given factor's {\it smallest\/} period is $p$; just that $p$ is one of the possible
periods.

\section{Prefixes}

In this section we prove the following theorem:

\begin{theorem}
Every nonempty prefix of the Thue-Morse word has circular critical exponent
in $S := \{ 1, 2, {7 \over 3}, {5 \over 2}, {{13} \over 5}, {8 \over 3}, 3 \}$.
\label{two}
\end{theorem}

Furthermore, we will precisely characterize the $n$ for which the circular critical
exponent is each member of $S$.

We start by creating a first-order formula asserting
that the length-$n$ prefix, considered as
a circular word, has some factor of length $m$ and (not necessarily least) period $p$,
satisfying $m/p = a/b$:
$$\prefgeab(n) := \exists i,m,p \ (p\geq 1) \, \wedge \, (m \leq n) \,
	\wedge \, (i<n) \, \wedge \, (bm \geq ap) \, \wedge \, \crep(i,m,n,p,0).
$$
Note that the condition $p \leq n$ need not be included explicitly, as it is
implied by the conjunction of $m \leq n$ and $bm \geq ap$.

Next, we create a formula asserting that the length-$n$ prefix, considered as
a circular word, has a factor with exponent $> a/b$:
$$ \prefgtab(n) := \exists i,m,p \ (p\geq 1) \, \wedge \, (m \leq n) \,
        \wedge \, (i<n) \, \wedge \, (bm > ap) \, \wedge \, \crep(i,m,n,p,0).
$$

Finally, we create a formula asserting that the length-$n$ prefix 
has some factor of exponent exactly $a/b$:
$$ \prefeqab(n) := \prefgeab(n) \, \wedge \, \neg \prefgtab(n) .$$

%
%
%
%
%
%

No single {\tt Walnut} command can be the direct translation of the formulas above,
as there is no way to take arbitrary integer
parameters $a,b$ as input and perform multiplication by them.  Nevertheless,
since there are only finitely many possibilities, we can translate the above
logical statements to finitely many
{\it individual\/} {\tt Walnut} commands for {\it each\/} exponent $a/b$.
For example, for $7/3$ we can write

\medskip

\noindent{\tt def prefge73 "E i,m,p (p>=1) \& (m<=n) \& (i<n) \& (3*m=7*p) \& \$crep(i,m,n,p,0)":}\\
\noindent{\tt def prefgt73 "E i,m,p (p>=1) \& (m<=n) \& (i<n) \& (3*m>7*p) \& \$crep(i,m,n,p,0)":} \\
\noindent{\tt def prefeq73 "\$prefge(n) \& \~\$prefgt73(n)":}

\medskip
\noindent and similarly for the other exponents.

\begin{proof}
We can now prove Theorem~\ref{two} by executing the {\tt Walnut} command

\medskip

\noindent {\tt eval testpref "An (n>=1) => (\$prefeq11(n) | \$prefeq21(n) | \$prefeq73(n) | } \\
{\tt \$prefeq52(n) | \$prefeq135(n) | \$prefeq83(n) | \$prefeq31(n))":}

\medskip

\noindent (where $|$ represents OR), and verifying that {\tt Walnut} returns {\tt true}.
\end{proof}

Table~\ref{tone} gives information about the state sizes of the automata for
the exponents in $S := \{ 1, 2, {7 \over 3}, {5 \over 2}, {{13} \over 5}, {8 \over 3}, 3 \}$.

In fact, even more is true.  We can create a {\it single} $2$-DFAO that on input
$n$ outputs the circular critical exponent of the prefix of length $n$ of
$\bf t$.  We do this by computing the automaton for each of the
possible exponents, forming the cross product automaton,
and producing the appropriate output.

\begin{table}[H]
\begin{center}
\begin{tabular}{ccccl}
      & number of  & number of & number of &  \\
$a/b$ & states for & states for & states for & first few $n$ accepted by prefeqab  \\
      & prefgeab      & prefgtab     & prefeqab & \\
\hline
1/1 & 2 & 4 & 3 & 1, 2\\
2/1 & 4 & 4 & 4 & $ 3, 4, 6, 8, 12, 16, 24, 32, 48, 64, 96, \ldots $ \\
7/3 & 4 & 12 & 7 & $13, 26, 37, 52, 61, 74, 93, \ldots$ \\
5/2 & 12 & 10 & 8 & $5, 10, 20, 29, 40, 45, 58, 77, 80, 90, \ldots$ \\
13/5 & 10 & 12 & 9 & $17, 34, 53, 65, 68, 85, \ldots$ \\
8/3 & 12 & 8 & 14 & $9, 18, 21, 27, 33, 36, 42, 43, 49, 54, \ldots $ \\
3/1 &  8 & 1 & 8 & $7, 11, 14, 15, 19, 22, 23, 25, 28, 30, \ldots $ \\
\end{tabular}
\end{center}
\label{tone}
\caption{State sizes for repetition of prefixes}
\end{table}

\begin{theorem}
There is a $2$-DFAO of 29 states that, on input $(n)_2$, returns
the circular critical exponent of the prefix of length $n$ of
$\bf t$.
\label{dfao29}
\end{theorem}

\begin{proof}
We cannot compute this automaton directly in {\tt Walnut} in its current
version, but it can be computed easily from the individual automata {\tt Walnut}
computes for each exponent in
$S = \{ 1, 2, {7 \over 3}, {5 \over 2}, {{13} \over 5}, {8 \over 3}, 3 \}$.

Now we can finish the (computational) proof of Theorem~\ref{dfao29}.
We start with the automaton {\tt prefeq11} discussed above.
Next,
for each of the remaining exponents $a/b$, we iteratively
form the cross product of the current automaton with 
the automaton {\tt prefeqab}, and remove unreachable states.
After all exponents are handled,
this gives the $29$-state automaton depicted
in Figure~\ref{figp}.

\begin{figure}[H]
\begin{center}
\includegraphics[width=6.5in]{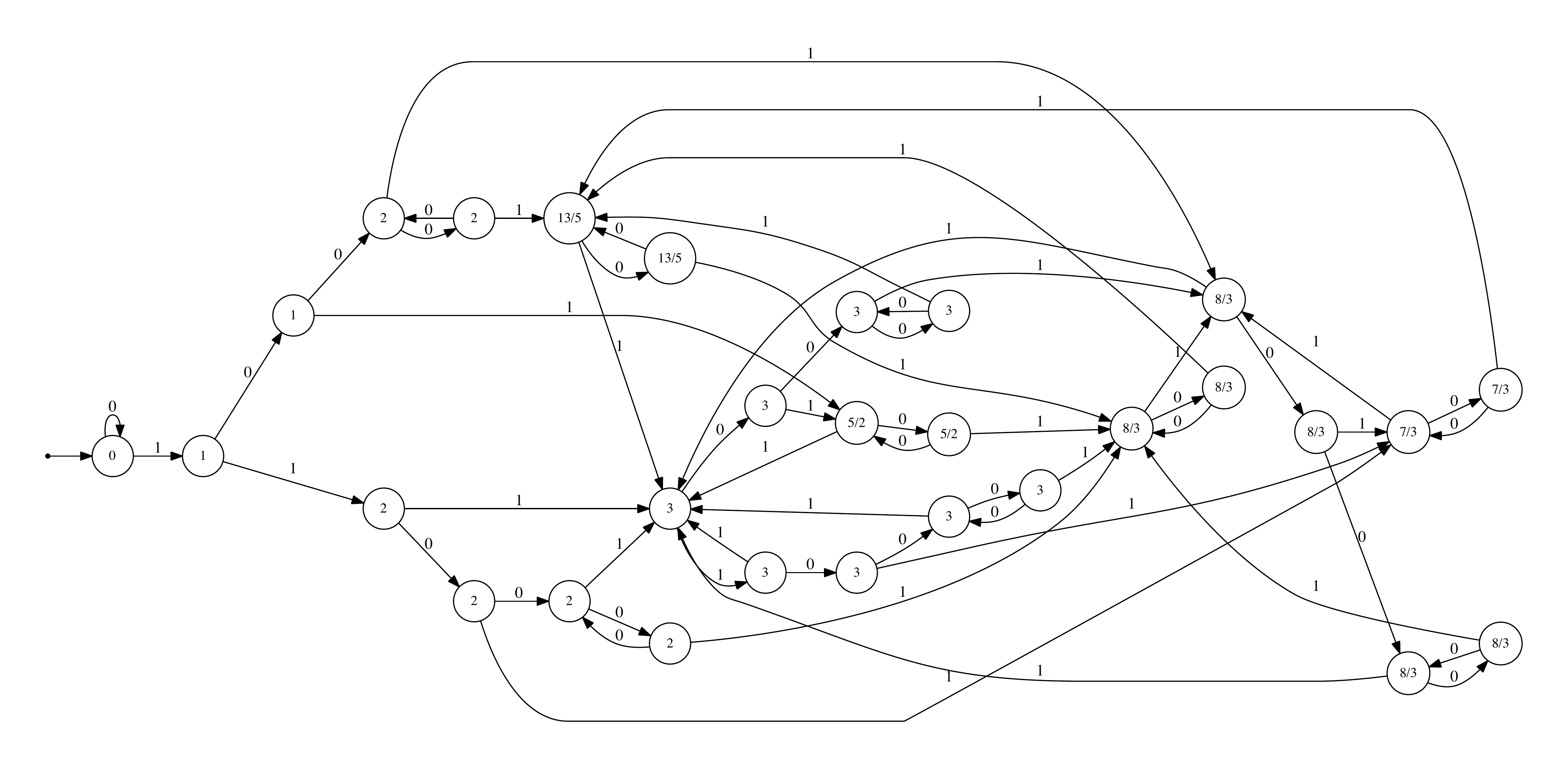}
\caption{Automaton for prefixes of Thue-Morse}
\label{figp}
\end{center}
\end{figure}
\end{proof}

\begin{remark}
We tested the automaton in Figure~\ref{figp} by explicitly calculating
$\ccexp$ for the first 500 prefixes of $\bf t$ and comparing the results.
They agreed in every case.
\end{remark}
 
\section{Factors}

Instead of just prefixes, we can carry out the calculations of the 
previous section for all factors.  The goal is to prove the following
result.
\begin{theorem}
Every factor of the Thue-Morse word has circular critical exponent lying in
the finite set $U := \{ 1, 2, {7 \over 3}, {{17} \over 7}, {5 \over 2},
{{13} \over 5}, {8 \over 3}, 3, {{10} \over 3}, {7 \over 2}, 
{{11} \over 3}, 4 \}$.
\end{theorem}

\begin{proof}
We can mimic the previous analysis.  A length-$n$ factor
${\bf t}[s..s+n-1]$ can be specified by the pair $(n,s)$.

We first make the assertion that the factor specified by $(n,s)$, considered
as a circular word, has a factor of length $m$ 
that has a period $p$ with $m/p = a/b$:
$$\facgeab(n) = \exists i,m,p \, (p\geq 1) \, \wedge \, (m \leq n) \,
\wedge \, (i \geq s) \,\wedge\, (i<s+n) \, \wedge \, (bm \geq ap) \, \wedge \, \crep(i,m,n,p,s).
$$


Next, we make the assertion that ${\bf t}[s..s+n-1]$, considered as
a circular word, has a factor with exponent $> a/b$:
$$ \facgtab(n) = \exists i,m,p \, (p\geq 1) \, \wedge \, (m \leq n) \,
        \wedge \, (i \geq s) \,\wedge (i<s+n) \, \wedge \, (bm > ap) \, \wedge \, \crep(i,m,n,p,s).
$$


Finally, we make the assertion ${\bf t}[s..s+n-1]$, considered as a 
circular word, has a factor of
exponent exactly $a/b$ and no larger:
$$ \faceqab(n) = \facgeab(n) \, \wedge \, \neg \facgtab(n) .$$

\begin{table}[H]
\begin{center}
\begin{tabular}{ccccc}
      & number of  & number of & number of &  first occurrence \\
$a/b$ & states for & states for & states for & $(n,s)$ of factor \\
      & facgeab      & facgtab     & faceqab  & with $\ccexp = a/b$ \\
\hline
1/1 & 2 & 6 & 5 & (1,0) \\
2/1 & 6 & 9 & 10 & (2,1) \\
7/3 & 9 & 50 & 43 & (7,3)\\ 
17/7 & 50 & 51 & 21 & (23,19) \\
5/2 & 51 & 71 & 41 & (5,0)  \\
13/5 & 71 & 63 & 33 & (13,8) \\
8/3 & 63 & 36 & 59  & (9,0) \\
3/1 & 36 & 24 & 35 & (4,1)  \\
10/3 & 24 & 26 & 15  & (10,3) \\
7/2 & 26 & 22 & 14  & (7,10)\\
11/3 & 22 & 21 & 16  & (18,3) \\
4/1 & 21 & 1 & 21 & (6,5)
\end{tabular}
\end{center}
\caption{State sizes for automata for circular exponents of factors}
\end{table}

Now we just make the assertion that one of the 12 possibilities always occurs:

\medskip

\noindent {\tt 
eval testfac "An (n>=1) => (As (\$faceq11(n,s) | \$faceq21(n,s) | \\
\$faceq73(n,s) | \$faceq177(n,s) | \$faceq52(n,s) | \$faceq135(n,s) | \\
\$faceq83(n,s) |  \$faceq31(n,s) | \$faceq103(n,s) | \$faceq72(n,s) | \\
\$faceq113(n,s) | \$faceq41(n,s)))":}


\medskip

\noindent and {\tt Walnut} evaluates it to be true.  Furthermore, it is easy to check
that each possibility occurs at least once, as given in the table.

\end{proof}

\begin{theorem}
There is a 204-state $2$-DFAO that, on input $(n,s)$ in base $2$, outputs
$$\ccexp({\bf t}[s..s+n-1]).$$
\end{theorem}

\begin{proof}
As before, we use the product construction to combine
the automata for $\faceqab(n,s)$ for all twelve possibilities
for $a/b$.  The automaton is too large to display here, but it
is available at \codelink \, .
\end{proof}

\begin{remark}
We tested the correctness of our automaton by comparing its result to
the result of thousands of randomly-chosen factors of
varying lengths of $\bf t$.  It passed all tests.
\end{remark}

\subsection{Smallest circular critical exponents for each length}

For every length $n$, we can consider the least circular 
critical exponent over all
factors ${\bf t}[s..s+n-1]$ of length $n$.  Define
$$ \lcce(n) = \min_{{x \text{ a factor of } {\bf t}} \atop {|x|=n}} \ccexp(x) .$$

\begin{theorem}
For all $n \geq 1$ we have
$\lcce(n) \in T$ where
$T := \{ 1, 2, {7 \over 3}, {{17} \over 7}, {5 \over 2} \}$.
\end{theorem}

\begin{proof}
First, we create a first-order logic statement asserting that
there exists some length-$n$ factor whose circular exponent equals $a/b$:
$$ \facab(n) := \exists s \ \faceqab(n,s) .$$

Next, we create a statement asserting that $a/b$ is the least circular
critical exponent for words of length $n$; in other words,
that there exists some
length-$n$ factor whose circular critical exponent equals $a/b$, and
furthermore every length-$n$ factor has circular critical exponent $\geq a/b$:
$$ \facsmallab(n) := \facab(n) \ \wedge \ (\forall s\ 
	\facgeab(n,s)) .$$

Finally, we just assert that for every $n \geq 1$, at least one of the
five alternatives holds:

\medskip

\noindent{\tt eval smallfactest "An (n >=1) => (\$facsmall11(n) | \$facsmall21(n) | \\
\$facsmall73(n) | \$facsmall177(n) | \$facsmall52(n))":}

\medskip

\noindent {\tt Walnut} evaluates this to be true.
\end{proof}


The sizes of the automata occurring in the proof are summarized below.

\begin{table}[H]
\begin{center}
\begin{tabular}{cccl}
      & number of  & number of &  \\
$a/b$ & states for & states for & first few $n$ accepted by $\facsmallab$ \\
      & facab & facsmallab  & \\
\hline
1/1 &  3 &  3 & $1, 2$ \\
2/1 &  3 & 4 & $3,4,6,8,12,16,24,32,48$ \\
7/3 &  20 & 20 & $7,13,14,19,21,25,26,27,28,29,33,35,37,38,42,43,45,49,50$ \\
17/7 & 7 & 7 & $23,31,39,46,47$ \\
5/2 &  9 & 16 & $5,9,10,11,15,17,18,20,22,30,34,36,40,41,44$
\end{tabular}
\end{center}
\caption{State sizes for $\facab$ and $\facsmallab$}
\end{table}

\begin{theorem}
There is a 25-state $2$-DFAO, that on input $(n)_2$ computes the least
circular critical exponent over all factor of $\bf t$ of length $n$.
It is given in Figure~\ref{smallfac}.
\end{theorem}

\begin{figure}[H]
\begin{center}
\includegraphics[width=6.5in]{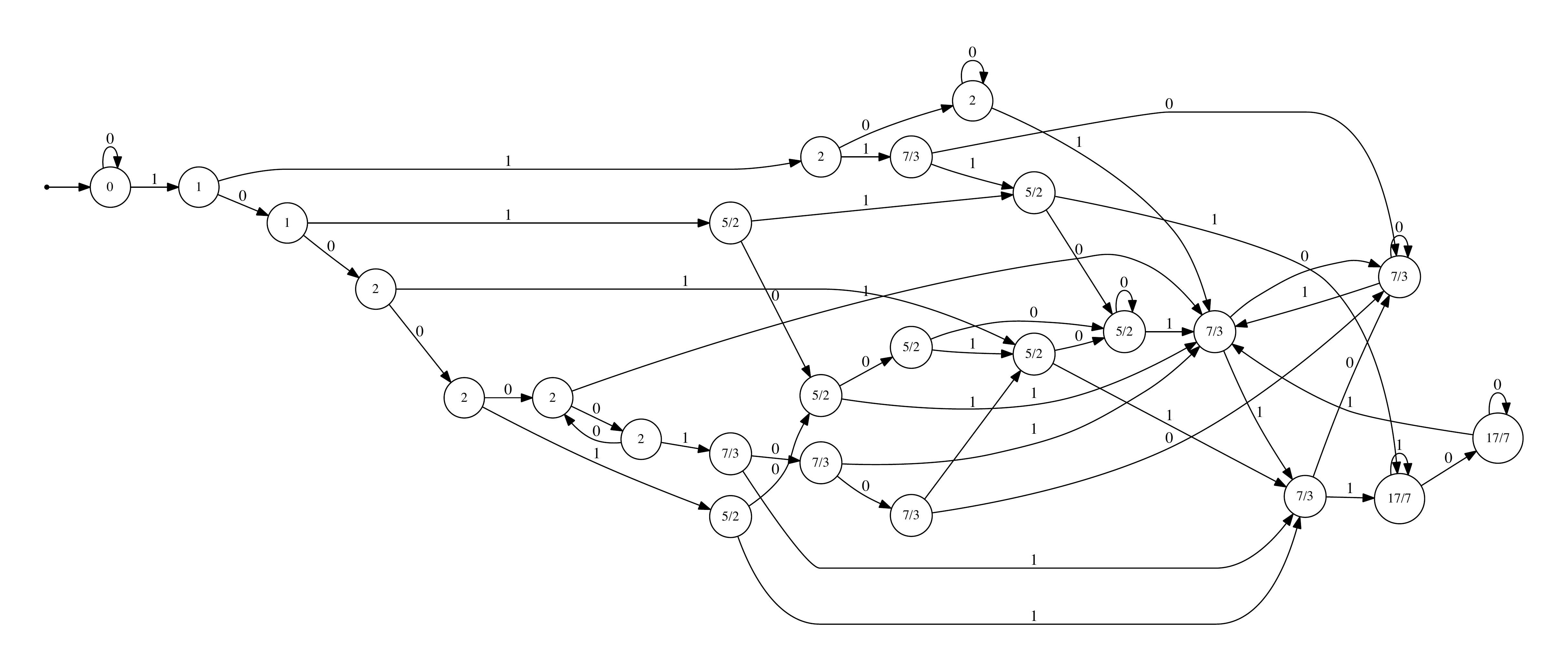}
\caption{Automaton computing least possible $\ccexp$ of a factor, for each length }
\label{smallfac}
\end{center}
\end{figure}

\begin{proof}
Combine, using the cross product construction, the five automata
$\facsmallab$ for $a/b \in \{1/1, 2/1, 7/3, 17/7, 5/2 \}$ 
as before.  The output of each state is depicted in the center of the
corresponding circle.
\end{proof}

\subsection{Greatest circular critical exponents for each length}

We can also consider the greatest circular critical exponent over all
factors ${\bf t}[s..s+n-1]$ of length $n$.  Define
$$ \gcce(n) = \max_{{x \text{ a factor of } {\bf t}} \atop {|x|=n}} \ccexp(x) .$$

\begin{theorem}
For all $n \geq 1$ we have
$\gcce(n) \in V$ where
$V := \{ 1, 2, 3, {7 \over 2}, 4 \}$.
\end{theorem}

\begin{proof}
We define
$$ \faclargeab(n) = (\exists s\ \faceqab(n,s)) \ \wedge \ 
	(\forall s\ \neg \facgtab(n,s)) .$$


The number of states, and the first few $n$ that match the category, are
given in Table~\ref{larget}.

We then verify the claim by writing

\medskip

\noindent{\tt eval largefactest "An (n>=1) => (\$faclarge11(n) | \$faclarge21(n) | \\
\$faclarge31(n) | \$faclarge72(n) | \$faclarge41(n))":}


\medskip

\noindent which returns {\tt true}.

\end{proof}

\begin{table}[H]
\begin{center}
\begin{tabular}{ccl}
      & number of  & \\
$a/b$ & states for & first few $n$ matching the case \\
      & faclargeab & \\
\hline
1/1 &  2 & 1  \\
2/1 &  3 & 2,3\\
3/1 &  8 & 4,5,9,13,15,17,21,25,29,33,37,41,45,49,53,57,61, \ldots \\
7/2 &  7 & 7,11,19,23,27,31,35,39,43,47,51,55,59,63, \ldots \\
4/1 &  5 & 6,8,10,12,14,16,18,20,22,24,26,28,30,32, \ldots
\end{tabular}
\end{center}
\caption{State sizes for $\faclargeab$}
\label{larget}
\end{table}

\begin{theorem}
There is a 9-state $2$-DFAO,
that on input $(n)_2$, returns the greatest circular
critical exponent over all length-$n$ factors of $\bf t$.
\end{theorem}

\begin{proof}
We follow the same approach as before, using the cross product construction
to combine the automata $\faclargeab$ for
$a/b \in \{1, 2, 3, {7 \over 2}, 4 \}$.  The result is depicted in
Figure~\ref{large}.
\end{proof}

\begin{figure}[H]
\begin{center}
\includegraphics[width=6.5in]{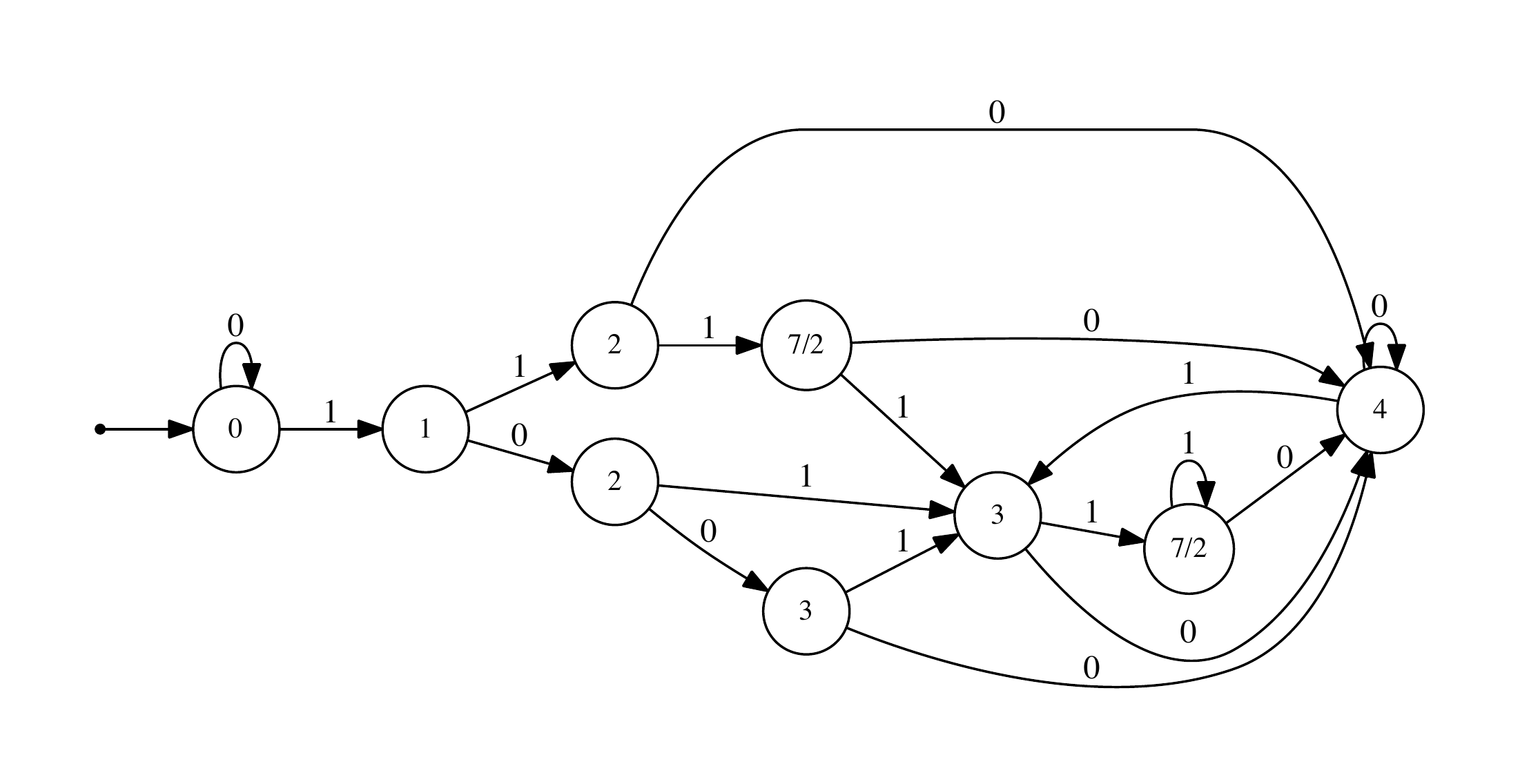}
\caption{Automaton computing greatest possible $\ccexp$ of a factor, for each length }
\label{large}
\end{center}
\end{figure}

\subsection{Sets of circular critical exponents}

We can get even more!  Define the set of all possible circular critical exponents of
factors of length $n$ as follows:
$$ \ace (n) = \{ \ccexp(x) \ : \ |x| = n \geq 1 \text{ and } x \text{ is a factor of {\bf t} } \} .$$

\begin{theorem}
The range of $\ace(n)$ consists of exactly $31$ distinct sets as
enumerated in Table~\ref{soce}.
\end{theorem}

\begin{proof}
This follows immediately from our proof of the next result.
\end{proof}

\begin{theorem}
There is a 49-state $2$-DFAO that, on input $n$ written in base $2$, 
outputs $\ace(n)$.
\end{theorem}

\begin{proof}
We use the same cross product automaton technique as before.  This time,
we use the automata $\facab$ for each $a/b \in U$.  The result is depicted in
Figure~\ref{big}.
The outputs associated with each state are encoded as $12$-bit
numbers, one for each of the 12 possible exponents in increasing order,
with least significant bit corresponding to exponent $4$.
Square states are ``transient'' and circular states are ``recurrent''.
\end{proof}

\section{Final remarks}

Evidently one could (in principle)
perform the same sort of analysis for many other famous infinite words.
We carried this out for the
regular paperfolding word 
$${\bf p} = 00100110001101100010 \cdots$$
(see, for example,
\cite{MendesFrance&Tenenbaum:1981,Dekking&MendesFrance&vanderPoorten:1982}),
and the results
are summarized below.  We omit the details, but the {\tt Walnut} code
proving these results is available at \codelink \, .  The computations were nontrivial.
{\tt Walnut} was invoked using the Linux command

\medskip

\centerline{\tt  java -Xmx16000M -d64 Main.prover} 

\medskip

\noindent on a 4 CPU AMD Opteron 6380 SE with 256GB RAM.  The analogue of
$\crep$ for $\bf p$ has 4226 states and took 9 minutes to compute. The largest 
intermediate automaton had 822,161 states.

\begin{theorem}
\leavevmode
\begin{itemize}
\item[(a)]  Every nonempty prefix of $\bf p$ has circular critical exponent
lying in $\{ 1,2,{7 \over 3}, 3, {{10} \over 3},4, {{13} \over 3}, 5 \}$.

\item[(b)]  Every nonempty factor of $\bf p$ has circular critical exponent
lying in \\ $\{1, 2, {7 \over 3}, {5 \over 2}, {8 \over 3}, {{11} \over 4}, 3,
{{10} \over 3}, {7 \over 2}, 4, {{13} \over 3}, 5, 6 \}$.

\item[(c)] The least circular critical exponent of $\bf p$, over all factors of
length $n$, lies in \\
$\{1, 2, {7 \over 3}, {5 \over 2}, {8\over 3}, {{11} \over 4}, 3\}$.

\item[(d)] The greatest circular critical exponent of $\bf p$, over all factors of
length $n$, lies in $\{ 1, 2, 3, 4, 5, 6 \}$.

\item[(e)]  There are exactly 16 distinct possible sets of circular
critical exponents for factors of length $n \geq 1$ of $\bf p$.
\end{itemize}
\end{theorem}

In principle, we could also treat the Rudin-Shapiro sequence.
For example, one might be able to prove
the following.

\begin{conjecture}
Every nonempty factor of the Rudin-Shapiro sequence has a circular critical
exponent lying in
$$\{ 1, 2, {5 \over 2}, {8 \over 3}, 3, {{10} \over 3}, {7 \over 2},
{{11} \over 3}, {{15} \over 4}, 4, {{21} \over 5}, {{13} \over 3},
{{14} \over 3}, 5, 6, 7, 8 \}.$$
\end{conjecture}
\noindent However, so far we have not been able to complete the computations
with {\tt Walnut} (it runs out of space).

For some infinite words, the sets under
consideration will be infinite, and hence another kind of analysis will be needed.
As an example, consider the infinite word $210201 \cdots$ that is a fixed point of
$2 \rightarrow 210$, $1 \rightarrow 20$, $0 \rightarrow 1$.  It is well-known that
this word is squarefree, but contains factors with exponent arbitrarily close to $2$.
In this case there will be no finite analogue of our Proposition~\ref{prop1} and 
Theorem~\ref{two}.
The same case occurs for the Fibonacci word (the fixed point of $0 \rightarrow 01$ and
$1 \rightarrow 0$).

\section{Acknowledgments}

The first author acknowledges, with thanks,
conversations with Daniel Go\v{c} on the subject in 2013.

\begin{table}[H]
\setlength{\extrarowheight}{0.1cm}
\begin{center}
\begin{tabular}{c|c|c}
set of circular & encoding & first few $n$ for \\
critical exponents $S$ & in automaton &  $\ace(n) = S$ \\
\hline
 $\{1\}$ & 2048 & $\{1\}$  \\
 $\{1,2\}$ & 3072 & $\{2\}$ \\
 $\{2\}$ & 1024 &  $\{3\}$ \\
 $\{2,3\}$ & 1040 & $\{4\}$ \\
 $\{{5 \over 2}, 3\}$ & 144 & $\{5\}$ \\
 $\{ 2,4 \}$ & 1025 & $\{ 6 \}$ \\
 $\{ {7 \over 3}, 3, {7 \over 2} \}$ & 532 & $\{7\}$ \\
 $\{ 2,3,4 \}$ & 1041 & $\{ 8,12,16,24,32,48,64,96,128,192, \ldots \}$ \\
 $\{ {5 \over 2}, {8 \over 3}, 3 \}$ & 176 & $\{ 9, 15 \}$ \\
 $\{ {5 \over 2}, 3, {{10} \over 3}, 4 \}$ & 153 & $\{10, 20, 40,80, 160, \ldots \}$  \\
 $\{ {5 \over 2}, 3, {7 \over 2} \}$ & 148 & $\{ 11 \}$ \\
 $\{ {7 \over 3}, {{13}\over 5}, {8 \over 3}, 3 \}$ & 624 & $\{ 13 \}$ \\
 $\{ {7 \over 3}, 3, {7 \over 2}, 4 \}$ & 533 & $\{ 14,28,56,112,224, \ldots \}$ \\
 $\{ {5 \over 2}, {{13} \over 5}, {8 \over 3}, 3 \}$ & 240 & $\{ 17,41,137,\ldots \}$ \\
 $\{ {5 \over 2}, {8 \over 3}, 3, {{11} \over 3}, 4 \} $ & 179 & $\{18,30,36,60,72,120,144, \ldots \}$ \\
 $\{ {7 \over 3}, {8 \over 3}, 3, {7 \over 2} \}$ & 564 & $\{ 19,67, \ldots \} $ \\
 $\{ {7 \over 3}, {5 \over 2}, {8 \over 3}, 3 \}$ & 688 & $\{ 21 \}$ \\
 $\{ {5 \over 2}, 3, {7 \over 2}, {{11} \over 3}, 4 \} $ & 151 & $\{ 22,44,88,176, \ldots \}$ \\
 $\{ {{17} \over 7}, {5 \over 2}, 3, {7 \over 2} \} $ & 404 & $\{ 23,71, \ldots \}$ \\
 $\{ {7 \over 3}, {5 \over 2}, {{13} \over 5}, {8 \over 3}, 3 \}$  & 752 & 
 	$\{ 25,29,33,37,45,49,53,57,61,65,69,73,77, \ldots \}$ \\
 $\{ {7 \over 3}, {{13} \over 5}, {8 \over 3}, 3, {{10} \over 3}, 4 \} $ & 633 & 
 	$\{ 26,52,104,208, \ldots \}$ \\
 $\{ {7 \over 3}, {5 \over 2}, {8 \over 3}, 3, {7 \over 2} \}$ & 692 & 
 	$\{ 27,35,43,51,59,75,83,91,99,107,115,123, \ldots \}$ \\
 $\{ {{17} \over 7}, {5 \over 2}, {8 \over 3}, 3, {7 \over 2} \}$ & 436 & 
 	$\{ 31,39,47,55,63,79,87,95,103,111,119,127, \ldots \}$  \\
 $\{ {5 \over 2}, {{13} \over 5}, {8 \over 3}, 3, {{10} \over 3}, {{11} \over 3}, 4 \}$ & 
  251 & $\{ 34,68,82,136,164, \ldots \}$ \\
 $\{ {7 \over 3}, {8 \over 3}, 3, {7 \over 2}, 4 \}$ & 565 & $\{ 38,76,134,152, \ldots \}$ \\
 $\{ {7 \over 3}, {5 \over 2}, {8 \over 3}, 3, {{10} \over 3}, {{11} \over 3}, 4 \} $ & 699 &
 	$\{ 42,84,168, \ldots \}$ \\
 $\{ {{17} \over 7}, {5 \over 2}, 3, {7 \over 2}, {{11} \over 3}, 4 \}$ & 407 & 
 	$\{46,92,142,184, \ldots \}$ \\
 $\{ {7 \over 3}, {5 \over 2}, {{13} \over 5}, {8\over 3}, 3, {{10} \over 3},
{{11} \over 3}, 4 \}$ & 763 & $\{ 50,58,66,90,98,100,106,114,116,122, \ldots \}$ \\
 $\{ {7 \over 3}, {5 \over 2}, {8\over 3}, 3, {7 \over 2}, {{11} \over 3}, 4 \} $ & 695 & 
 	$\{ 54,70,86,102,108,118, \ldots \}$  \\
 $\{ {{17} \over 7}, {5 \over 2}, {8 \over 3}, 3, {7 \over 2}, {{11} \over 3}, 4 \}$ & 
 	439 & $\{62,78,94,110,124,126, \ldots \}$ \\
 $\{ {7 \over 3}, {5 \over 2}, {{13} \over 5}, {8 \over 3}, 3, {{10} \over 3}, 4 \}$ & 
 	761 & $\{74,148, \ldots \}$
\end{tabular}
\end{center}
\caption{Sets of circular critical exponents for lengths $n$}
\label{soce}
\end{table}

\begin{landscape}

\begin{figure}[H]
\begin{center}
\includegraphics[width=7.5in]{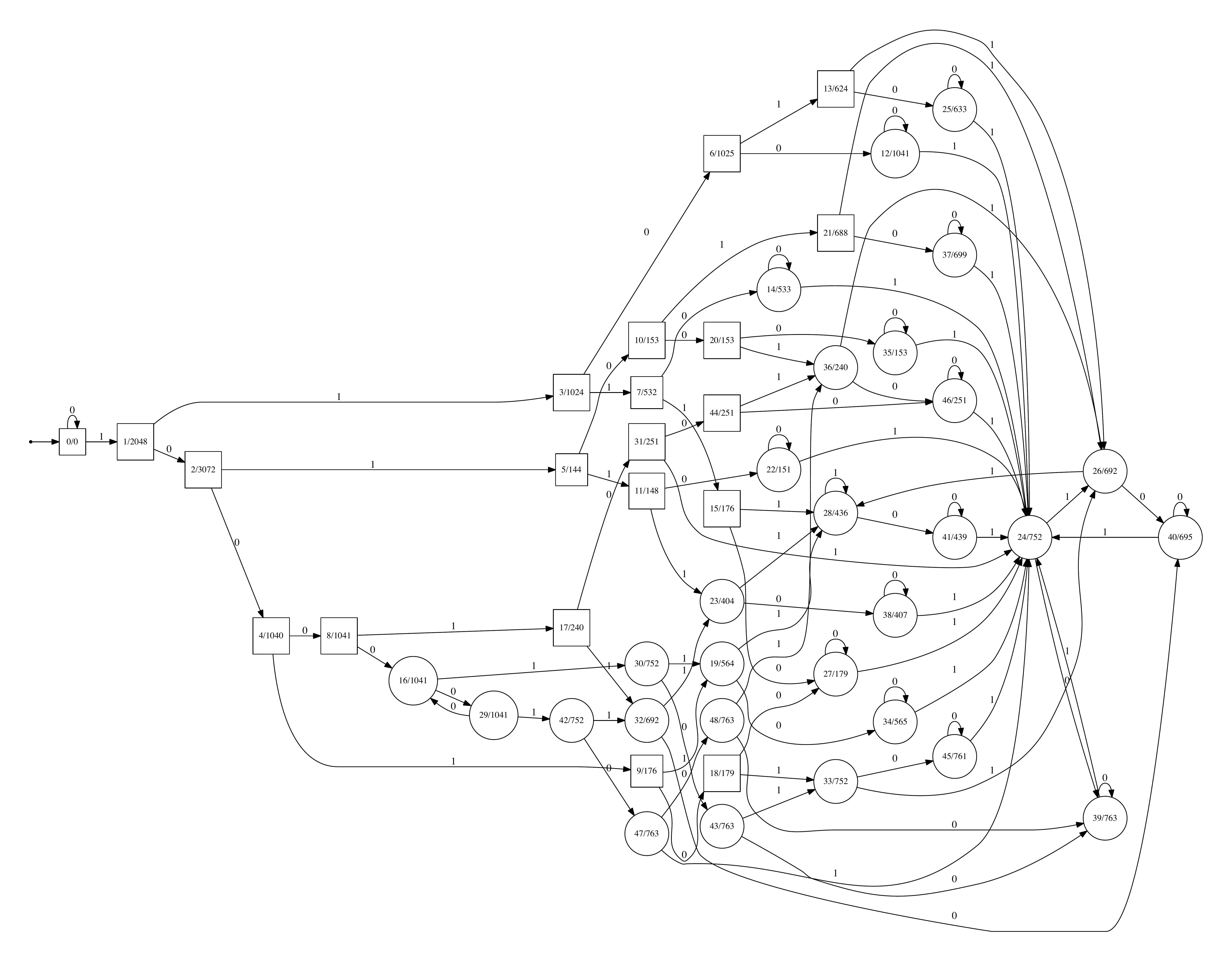}
\caption{Automaton computing sets of circular critical exponents for
factors of length $n$}
\label{big}
\end{center}
\end{figure}

\end{landscape}

\end{document}